\documentclass[letterpaper]{article} 
\usepackage{aaai24}  
\usepackage{times}  
\usepackage{helvet}  
\usepackage{courier}  
\usepackage[hyphens]{url}  
\usepackage{graphicx} 
\usepackage{natbib}  
\usepackage{caption} 
\usepackage{algorithm}
\usepackage{algorithmic}
\usepackage{amsmath,amssymb, amsthm}
\usepackage{mathabx}
\usepackage{soul}
\usepackage{pgfplots}
\pgfplotsset{width=10cm,compat=1.9}
\usepgfplotslibrary{external}
\usepackage{amsopn}

\newtheorem{theorem}{Theorem} 
\newtheorem{lemma}[theorem]{Lemma}
\newtheorem{definition}[theorem]{Definition}
\newtheorem{proposition}[theorem]{Proposition}

{Assumption} 

\newtheorem{corollary}[theorem]{Corollary}

\newtheorem*{definition*}{Definition}

\newcommand{\ZZ}{\mathcal{Z}}

\newcommand{\eps}{\varepsilon}

\newcommand{\PP}{\mathcal{P}}

\usepackage{cleveref}
\crefname{algorithm}{Algorithm}{Algorithms}
\crefname{assumption}{Assumption}{Assumptions}
\crefname{equation}{}{}
\crefname{figure}{Fig.}{Figs.}
\crefname{table}{Table}{Tables}
\crefname{section}{Section}{Sections}
\crefname{subsection}{Section}{Sections}
\crefname{theorem}{Theorem}{Theorems}
\crefname{lemma}{Lemma}{Lemmmas}
\crefname{proposition}{Proposition}{Propositions}
\crefname{definition}{Definition}{Definitions}
\crefname{corollary}{Corollary}{Corollaries}
\crefname{remark}{Remark}{Remarks}
\crefname{example}{Example}{Examples}
\crefname{appendix}{Appendix}{Appendices}

\usepackage{amsmath}
\usepackage{amsfonts}
\usepackage{color}
\usepackage{graphicx}
\usepackage{enumitem}
\usepackage{subfigure}
\usepackage{svg}
\usepackage{algorithm}
\usepackage{algorithmic}
\newcommand{\XX}{\mathcal{X}}

\newcommand{\WW}{\mathcal{W}}

\newcommand{\alg}{\mathcal{A}}
\usepackage{mathtools}

\newcommand{\Al}{\mathcal{A}}

\newcommand{\pr}{\mathbb{P}}
\newcommand{\prin}{\pr_{\text{in}, x}}
\newcommand{\prout}{\pr_{\text{out}, x}}

\newcommand{\hockey}{D_{e^{\eps}}}
\newcommand{\xin}{X_{\text{in}}(x)}
\newcommand{\xout}{X_{\text{out}}(x)}

\usepackage{newfloat}
\usepackage{listings}
\DeclareCaptionStyle{ruled}{labelfont=normalfont,labelsep=colon,strut=off} 
\floatstyle{ruled}
\newfloat{listing}{tb}{lst}{}
\floatname{listing}{Listing}
\title{Why Does Differential Privacy with Large $\eps$ Defend Against Practical Membership Inference Attacks?}

\author {
    Andrew Lowy\textsuperscript{\rm 1},
    Zhuohang Li\textsuperscript{\rm 2},
    Jing Liu\textsuperscript{\rm 3},
    Toshiaki Koike-Akino\textsuperscript{\rm 3},
    Kieran Parsons\textsuperscript{\rm 3},
    Ye Wang\textsuperscript{\rm 3}
}
\affiliations {
    \textsuperscript{\rm 1}University of Wisconsin-Madison\\
    \textsuperscript{\rm 2}Vanderbilt University\\
    \textsuperscript{\rm 3}Mitsubishi Electric Research Laboratories\\
    alowy@wisc.edu, zhuohang.li@vanderbilt.edu, \{jiliu, koike, parsons, yewang\}@merl.com
}

\begin{document}

\maketitle

\begin{abstract}
For ``small'' privacy parameter $\eps$ (e.g. $\eps < 1$), $\eps$-differential privacy (DP) provides a strong \textit{worst-case} guarantee that no \textit{membership inference attack} (MIA) can succeed at determining whether a person's data was used to train a machine learning model. The guarantee of DP is \textit{worst-case} because: a) it holds even if the attacker already knows the records of all but one person in the data set; and b) it holds uniformly over all data sets. In practical applications, such a worst-case guarantee may be overkill: \textit{practical} attackers may lack exact knowledge of (nearly all of) the private data, and our data set might be easier to defend, in some sense, than the worst-case data set. Such considerations have motivated the industrial deployment of DP models with large privacy parameter (e.g. $\eps \geq 7$), and it has been observed empirically that DP with large $\eps$ can successfully defend against state-of-the-art MIAs. Existing DP theory cannot explain these empirical findings: e.g., the theoretical privacy guarantees of $\eps \geq 7$ are essentially vacuous. In this paper, we aim to close this gap between theory and practice and understand \textit{why a large DP parameter can prevent practical MIAs}. To tackle this problem, we propose a new privacy notion called \textit{practical membership privacy} (PMP). PMP models a practical attacker's uncertainty about the contents of the private data. The PMP parameter has a natural interpretation in terms of the success rate of a practical MIA on a given data set. We quantitatively analyze the PMP parameter of two fundamental DP mechanisms: the exponential mechanism and Gaussian mechanism. Our analysis reveals that a large DP parameter often translates into a much smaller PMP parameter, which guarantees strong privacy against \textit{practical} MIAs. Using our findings, we offer principled guidance for practitioners in choosing the DP parameter. 
\end{abstract}

\section{Introduction}

Machine learning (ML) systems, such as large language models (LLMs), have the potential to transform various facets of society and industry. However, the growing ubiquity of these systems raises privacy concerns and a long line of work has demonstrated how to \textit{attack} ML models and uncover private details about individuals whose data was used to train the model. For example, \cite{carlini2021extracting} extracted individual training examples by querying an LLM. 

\textit{Membership inference attacks} (MIAs)~\cite{shokri2015privacy,dwork2015robust} are a fundamental class of privacy attacks. An MIA receives a trained model $\alg(D)$ and a target data point $x$ as inputs and aims to infer whether or not the target point was used to train the model (i.e. whether or not $x \in D$). In this paper, we focus on white-box attackers who know the (randomized) algorithm $\alg$. MIAs can violate people's privacy: for example, genomic data sets may contain information about people with a particular medical diagnosis, and knowing that someone is in the data set reveals that they have the diagnosis~\cite{homer2008resolving}. Moreover, MIAs are often used as building blocks for other attacks, such as training data extraction attacks~\cite{carlini2021extracting}. Thus, if we can prevent MIAs, we can often also prevent other attacks. 

For sufficiently small $\delta$ and $\eps \geq 0$, \textit{$(\eps, \delta)$-differential privacy} (DP)~\cite{dwork2006calibrating} guarantees that no MIA can succeed with high probability, by requiring that the output distribution of the ML model be insensitive to the presence or absence of any individual data point (see Definition~\ref{def: DP}). Pure $\eps$-DP bounds the probability that an arbitrary MIA can succeed by $1/(1 + e^{-\eps})$. Thus, for example, $\eps \leq 0.1$ implies that no MIA can do much better than randomly guessing ($52.5\%$) whether or not a target data point was used to train the model. However, the guarantee of DP degrades rapidly with $\eps$,
e.g., if $\eps \geq 7$, then an $\eps$-DP algorithm is potentially vulnerable to MIAs that succeed with probability $\geq 0.999$. On the other hand, large $\eps$ values of $\eps \geq 7$ are often deployed in industrial applications~\cite{apple-differential-privacy, rappor14, ding2017collecting, desfontainesblog20211001}. Moreover, values of $\eps \geq 8$ have been shown empirically to be highly effective at thwarting state-of-the-art MIAs~\cite{carlini2022membership}. Existing theory cannot adequately explain the empirical success of $\eps$-DP with large $\eps$ at defending against MIAs. 

This paper aims to bridge this gap between theory and practice. We rigorously address the following question:
\begin{center}
\textit{Why does DP with large $\eps$ defend against practical
MIAs?}   
\end{center}

\paragraph{Contributions}
To answer this question, we begin with the observation that differential privacy provides a guarantee against a \textit{worst-case} MIA, which holds uniformly for all data sets. Namely, \textit{DP ensures that even an attacker with knowledge of $n-1$ data points $D \setminus \{x\}$ cannot infer whether or not the target point $x$ was used as input to $\alg$}. On the other hand, \textit{practical} attackers typically do not have such fine-grained knowledge of the underlying data set as the worst-case attacker that DP models. Indeed, the literature on MIAs typically assumes that the attacker has some knowledge of the data \textit{distribution} (e.g. query access to the distribution or knowledge of a subpopulation from which the data was randomly drawn), but does not know any of the points in the given data set with certainty. The attacker must rely on the output of the training algorithm $\alg(D)$ and distributional knowledge to infer membership of the target point $x$.  

We model this practical MIA setting in our definition of \textit{practical membership privacy} (PMP, Definition~\ref{def: PMP}). We show that PMP is a useful notion of privacy: PMP is weaker than the strong worst-case notion of DP (Proposition~\ref{prop: PMP is weaker than DP}), but strong enough to guarantee that no practical MIA can succeed with high probability (Lemma~\ref{lem: attack success rate}). Moreover, PMP is not susceptible to the blatant privacy breaches that afflict other weakenings of DP that have been defined in the literature (see Related Work and Appendix). 

We analyze the relationship between the PMP and DP parameters for two popular DP mechanisms: the exponential mechanism~\cite{mcsherry2007mechanism} and the Gaussian mechanism~\cite{dwork2006calibrating}. We show that the PMP parameter can be much smaller than the DP parameter for these mechanisms, e.g., the $\eps$-DP exponential mechanism satisfies $\eps/75$-PMP for certain subpopulations. This helps explain why large values of $\eps$ can provide strong protection against practical MIAs: for example, the $(\eps = 7.5)$-DP exponential mechanism lacks meaningful privacy guarantees against a worst-case attacker, but the resulting $(\eps/75 = 0.1)$-PMP guarantee ensures that no practical MIA can succeed with probability much higher than random guessing ($52.5\%$). We conclude by discussing the implications of our results for practitioners in choosing the DP parameter, and highlighting interesting directions for future work.

\paragraph{Differential Privacy} 
\begin{definition}[Differential Privacy~\cite{dwork2006calibrating}]
\label{def: DP}
Let $\varepsilon \geq 0, ~\delta \in [0, 1).$ A randomized algorithm $\Al: \XX^n \to \mathcal{Z}$ is \textit{$(\varepsilon, \delta)$-differentially private} (DP) if for all pairs of adjacent data sets $D, D' \in \XX^n$
and all measurable subsets $S \subseteq \mathcal{Z}$, we have
\begin{equation*}
\mathbb{P}(\alg(D) \in S) \leq e^\varepsilon \mathbb{P}(\alg(D') \in S) + \delta,
\end{equation*}
where the probability is solely over the randomness of $\alg$. 
If~$\delta = 0$, we say that $\alg$ satisfies ``pure DP'' and write $\eps$-DP. If~$\delta > 0$, we say ``approximate DP'' and write $(\eps, \delta)$-DP. 
\end{definition}

\section{Practical Membership Privacy}

In this section, we define a privacy notion---called practical membership privacy (PMP)---that models the practical MIA setting.  

PMP models a membership inference attacker who does not know any elements of $D^*$ with certainty, but has some distributional knowledge of $D$. Specifically, we assume that the attacker knows \textit{a ``parent set'' $X \in \XX^{2n}$ from which $D$ was drawn uniformly at random}\footnote{The choice of $2n$ as the size of the parent set is for analytical convenience. Our analysis extends to the case where, e.g., $X$ contains $\alpha n$ points for some $\alpha > 1$.}. One can interpret the parent set $X$ as representing a subpopulation from which the data was known to be drawn (e.g., health insurance customers or hospital patients) or
a dataset (e.g., MNIST) consisting of training samples $D$ and test samples $X \setminus D$.
PMP ensures that such an attacker cannot succeed in correctly determining membership of any target point with high probability:

\begin{definition}[Practical Membership Privacy\footnote{To simplify some of our analyses, we will assume that $X$ consists of $2n$  distinct points, w.l.o.g: If there are repeated points, then we can re-define $X$  without repeats for some smaller $n$.}]
\label{def: PMP}
Let $\varepsilon \geq 0, ~\delta \in [0, 1)$ and $X \in \XX^{2n}$. A randomized algorithm $\Al: \XX^n \to \mathcal{Z}$ satisfies \textit{$(\varepsilon, \delta)$-practical membership privacy} (PMP) with respect to $X$ if for all $x \in X$ and all measurable subsets $S \subseteq \mathcal{Z}$, we have
\begin{align*}
&e^{-\varepsilon} \left(\mathbb{P}(\alg(D) \in S | x\notin D) - \delta\right) \\
&\quad \leq \mathbb{P}(\alg(D) \in S | x \in D) \\
&\quad \leq e^\varepsilon \mathbb{P}(\alg(D) \in S | x\notin D) + \delta,
\end{align*}
where the probability is taken both over the random draw of $D \sim \textbf{Unif}\left(\{E \subset X: |E| = n \}\right)$ and the randomness of $\alg$. $\alg$ is $(\eps, \delta)$-PMP if $\alg$ is $(\eps, \delta)$-PMP with respect to $X$ for all $X \in \XX^{2n}$.
To denote ``pure PMP'', when $\delta = 0$, we will simply write $\eps$-PMP as a shorthand for ($\eps, 0$)-PMP.
\end{definition}

The key differences between the PMP model and the DP model are: 1) our (practical) attacker only has partial information about the other $n-1$ samples in $D$, whereas DP allows the (worst-case) attacker to know the other $n-1$ samples with certainty; and 2) our definition is dependent on the parent data set $X$, whereas DP holds uniformly over all data sets. Our assumption on the attacker's knowledge is more realistic than the DP assumption in many private data analysis settings: In practice, it is uncommon that an attacker knows $n-1$ points in a data set (but not the $n$-th point). However, it is often the case that an attacker knows that the data set was drawn from some sub-population $X \subset \XX$; Definition~\ref{def: PMP} models an attacker with this knowledge.

The following lemma provides alternative characterizations of PMP: 
\begin{lemma}
\label{lem: equivalent definitions of PMP}
Let $X \in \XX^{2n}$, $x \in X$, $\xin := \{D \subset X : |D| = n, x \in X\}$, and $\xout = \{D \subset X : |D| = n, x \notin X\}$. Let $S \subset \ZZ$ be a measurable set. 
If
\begin{align}
\label{eq: a}
&e^{-\eps}\left(\pr(x \notin D | \alg(D) \in S) - \delta\right) \nonumber \\
&\quad \leq \pr(x \in D | \alg(D) \in S) \nonumber \\
&\quad \leq e^\eps \pr(x \notin D | \alg(D) \in S) + \delta,
\end{align}
then
\begin{align}
\label{eq: b}
&e^{-\eps}\left(\pr(\alg(D) \in S | x \notin D) - 2\delta\right) \nonumber \\
&\quad \leq \pr(\alg(D) \in S | x \in D ) \nonumber \\
&\quad \leq e^\eps \pr( \alg(D) \in S |x \notin D ) + 2\delta.
\end{align}
Also, \cref{eq: b} holds iff \begin{align}
\label{eq: c}
&e^{-\eps}\left(\frac{1}{N} \sum_{D' \in \xout} \pr_{\alg}(\alg(D') \in S) - \delta \right) \nonumber \\
&\quad \leq \frac{1}{N} \sum_{D \in \xin} \pr_{\alg}(\alg(D) \in S)  \nonumber \\
&\quad \leq e^\eps \left(\frac{1}{N} \sum_{D' \in \xout} \pr_{\alg}(\alg(D') \in S)\right) + \delta,
\end{align}
where $N := |\xin| = |\xout| = {2n \choose n}/2$ and the probabilities in \cref{eq: c} are taken solely over the randomness of $\alg$.

Moreover, if $\delta = 0$, then \cref{eq: a} holds iff \cref{eq: b} holds iff \cref{eq: c} holds. Thus, $\alg$ is $\eps$-PMP w.r.t. $X$ iff any of these three inequalities holds for all $x \in X$ and all $S \subset \mathcal{Z}$. 
\end{lemma}

Proofs are deferred to the Appendix. A consequence of 
the equivalence between \cref{eq: b} and \cref{eq: c} is that if $n=1$, then $\eps$-PMP and $\eps$-DP are equivalent---and satisfy $\eps$-\textit{local differential privacy}~\cite{whatcanwelearnprivately}:
\begin{corollary}
\label{coro: DP = PMP}
If $n=1$, then $\alg$ is $(\eps, \delta)$-DP iff $\alg$ is $(\eps, 2 \delta)$-PMP w.r.t. $X$ for every $X \in \XX^{2n}$.    
\end{corollary}

For $n>1$,  PMP is weaker than DP. For simplicity, we present this result for $\delta = 0$:
\begin{proposition}
\label{prop: PMP is weaker than DP}
If $\alg$ is $\eps$-DP, then $\alg$ is $\eps$-PMP. Moreover, if $n > 2$, then there exists an $\ln(2)$-PMP $\alg$ that is not $\eps'$-DP for any $\eps' < \infty$. 
\end{proposition}

Intuitively, Proposition~\ref{prop: PMP is weaker than DP} is true because the inequalities in~\cref{eq: c} involve averages over data points in $X$, rather than the worst-case supremum appearing in the definition of DP. Moreover, the data set might not be worst case for PMP. The averages correspond to the practical attacker's uncertainty about which samples are in $D$, which makes it harder to infer membership of $x$ than the worst-case DP attacker. Also, the $\ln(2)$-PMP parameter in Proposition~\ref{prop: PMP is weaker than DP} is not tight, as our construction for $n=3$ can be extended to get an $\eps$-PMP algorithm with $\eps < \ln(2)$ for $n > 3$, e.g., one can get $\eps \leq \ln(40/37) < 0.08$ for $n = 6$. 

Next, we bound the success probability of a \textit{practical MIA} (as defined at the beginning of this section) in terms of the PMP parameter: 
\begin{lemma}
\label{lem: attack success rate}
Let $\alg$ be $\eps$-PMP with respect to $X$ and $\mathcal{M}$ be any practical MIA. Then, the probability that $\mathcal{M}$ successfully infers membership, for any $x \in X$, never exceeds $1/(1 + e^{-\eps})$. 
\end{lemma}
Analogously, it is well-known that $\eps$-DP ensures that that success probability of the \textit{worst-case} attacker (who knows all but one sample of $D$) never exceeds $1/(1 + e^{-\eps})$.

In the Appendix, we record additional basic properties of PMP, such as post-processing.

\section{Related Work}
Some prior works have sought to understand why large $\eps$ effectively prevents practical privacy attacks from various different angles. Most of these approaches seek to weaken the assumptions on the attacker in some respect. 
See \citet[Section 7]{ghazi2022algorithms} for a thorough discussion of different directions in which weaker assumptions on the attacker may be imposed. 
Below, we list these directions and cite a few related works for each.

\paragraph{Assumptions about the attacker's capabilities:} DP assumes that the attacker has unlimited computational resources and is capable of executing any sort of attack. Some relaxations of DP, such as computational DP \cite{mironov2009computational}, model an attacker with limited computational resources. Other privacy notions (e.g., $k$-anonymity) model an attacker that only executes a specific type of attack (e.g., record-linkage attack). In contrast to these works, our PMP notion models an attacker with the same vast capabilities as the DP attacker.  

\paragraph{Assumptions about the attacker's goals:} DP protects against membership inference attacks, which is equivalent (up to a factor of $2$ in $\eps$) to an attacker learning an arbitrary one-bit function of the target individual's data. Some works have considered a modified attacker with more ambitious goals (e.g., training data reconstruction~\cite{hayes2023bounding}). Other works have relaxed the DP definition to consider an attacker that only aims to extract certain bits of information from the target individual, e.g., attribute-level partial DP~\cite{ghazi2022algorithms}. In contrast to these works, our work considers an attacker with the same goals as the DP attacker. Thus, the attacker that we model is stronger along the ``goals'' axis than these prior works. 

\paragraph{Assumptions about the attacker's knowledge:} DP permits an attacker to know everything about the data set except for one private bit that they aim to infer. Several works have sought to model the uncertainty that a practical attacker has about the contents of the data set, e.g.,~\cite{bassily2013coupled,li2013membership,yeom2018privacy,sablayrolles2019white,humphries2020investigating,izzo2022provable,leemann2023gaussian}. 

Similarly, our PMP notion models an attacker with weaker knowledge than the DP attacker. PMP has advantages over previously proposed privacy notions that model the attacker's uncertainty. For example, as we discuss in the Appendix, many previously proposed definitions can be satisfied by algorithms that leak the data of some members of the data set and are therefore not (intuitively) private. By contrast, PMP is not susceptible to these blatant privacy violations. Moreover, the focus of our work---on precisely understanding the risk of a privacy breach with a practical (uncertain) attacker against specific DP algorithms---is different from these prior works.

In the Appendix, we discuss prior works seeking to weaken assumptions about the attacker's knowledge in more detail. We highlight pathologies with previously proposed definitions, in which algorithms that clearly leak an individual's data can still satisfy these other definitions. Also, in contrast to some other works, PMP does not impose any distributional or independence assumptions on the underlying data. Instead, we allow for data to be drawn from an arbitrary subpopulation $X$. This makes our analysis harder, but also makes our definition and results stronger. Finally, we reiterate that prior works did not provide the quantitative interpretations of practical privacy guarantees of concrete DP mechanisms that our work provides. In this work, \textbf{we give quantitative bounds  relating the DP parameter $\eps$ to the PMP parameter} and a \textbf{precise interpretation of the guarantees of our PMP notion against any practical attacker} (Lemma~\ref{lem: attack success rate}). Together, these results enable a \textbf{rigorous interpretation of the privacy guarantees of $\eps$-DP against a practical (less knowledgeable) attacker}.

\section{Practical Privacy Guarantees of the Exponential Mechanism}
In this section, we characterize the practical membership privacy of one of the most powerful and versatile differentially private algorithms: the \textit{exponential mechanism}~\cite{mcsherry2007mechanism}. 
To define the exponential mechanism, let $\WW$ be a finite set of objects.\footnote{If $\WW$ is infinite, then the exponential mechanism can still be applied after discretizing $\WW$.} Let $\ell: \WW \times \XX^n \to \mathbb{R}$ be some loss function. Given data $D$, our goal is to privately select an object $w \in \WW$ that approximately minimizes the loss function.

\begin{definition}[Exponential Mechanism]
Given inputs $D, \WW, \ell$, the exponential mechanism $\alg_E$ selects and outputs some object $w \in \WW$. The probability that a particular $w$ is selected is proportional to $\exp\left(\frac{-\eps \ell(w, D)}{2 \Delta_\ell} \right)$, where $\Delta_\ell = \max_{w \in \WW} \sup_{D \sim D'; D, D' \in \XX^n} |\ell(w, D) - \ell(w, D')|$.
\end{definition}
 
\begin{lemma}~\cite{mcsherry2007mechanism}
    The exponential mechanism is $\eps$-DP.
\end{lemma}

The following proposition gives an exact description of the PMP parameter as a function of the DP parameter $\eps$: 
\begin{proposition}
\label{prop: exponential mechanism PMP}
Let $X \in \XX^{2n}$. The $\eps$-DP exponential mechanism is $\tilde{\eps}(X)$-PMP with respect to $X$ if and only if \[
\tilde{\eps}(X) \geq \ln\left[\frac{\sum_{D \in \xin} c(D) \exp\left(-\frac{\eps}{2 \Delta_\ell} \ell(w,D) \right)}{\sum_{D' \in \xout} c(D') \exp\left(-\frac{\eps}{2 \Delta_\ell} \ell(w,D')\right)} \right]
\]
for all $w \in \WW$ and $x \in X$,
where $c(D) = \left[\sum_{w' \in \WW} \exp\left(\frac{- \eps \ell(w', D)}{2\Delta_{\ell}} \right)\right]^{-1}$ and $c(D')$ is defined similarly. 
\end{proposition}

For a given loss function $\ell$ and subpopulation $X$, Proposition~\ref{prop: exponential mechanism PMP} allows us to compute the PMP paramater $\tilde{\eps}(X)$ of the exponential mechanism as a function of $\eps$. In combination with Lemma~\ref{lem: attack success rate}, this will allow us to interpret $\eps$ in terms of the success rate of an arbitrary practical membership inference attacker.

\paragraph{Numerical Simulations}
We investigate the PMP parameter $\tilde{\eps}$ vs. the DP parameter $\eps$ for different subpopulations $X$. We fix the loss function: $\ell(w, D) = \frac{1}{n} \sum_{i=1}^n \|w - D_i\|_2$, which is a convex empirical risk minimization problem corresponding to the geometric median. Our goal is to understand the \textit{ratio} $\tilde{\eps}(X)/\eps$ that we get for different $X$, and different factors that affect the ratio (e.g., the distribution and dimension of the data). We choose $\WW = \{w_1, \ldots, w_m\}$ to be a set of $m$ random standard normal unit vectors in $\mathbb{R}^d$, standardized to have unit $\ell_2$-norm. We then draw $X \sim \mathcal{N}(w_1, \sigma^2)^{2n \times d}$ and clip the $\ell_2$ norm of each data point, so $\|x_i\|_2 \leq C$ for all $i \in [2n]$, where $C$ is the clip threshold.  

Recall that there are two key differences between PMP and DP: one difference lies in the attacker's knowledge/uncertainty about the data, and the second is that PMP is defined with respect to a subpopulation $X$, whereas DP is worst-case over all $X$. In order to disentangle these two effects, we plot two curves in each experiment: the (average, over $T$ trials) ratios $\tilde{\eps}(X)/\eps$ and the ratio$_X$ $\tilde{\eps}(X)/\eps(X)$. Here $\eps(X)$ is defined as in Definition~\ref{def: DP} except that that we only require the inequality to hold for adjacent data sets $D, D'$ that are subsets of $X$, rather than $\XX^n = \{x \in \mathbb{R}^d : \|x\|_2 \leq C\}^n$.
The ratio$_X$ $\tilde{\eps}(X)/\eps(X)$ controls for the effect of the data and just describes the effect of the practical attacker's uncertainty compared to the worst-case DP attacker's certainty about members of $D \setminus \{x\}$. The ratio $\tilde{\eps}(X)/\eps$ captures the role of both the attacker's knowledge and the data being potentially easier to defend than the worst-case data set.  

Figure~\ref{fig: exp mech ratio v sigma} shows the ratios $\tilde{\eps}(X)/\eps$ and $\tilde{\eps}(X)/\eps(X)$ vs. the standard deviation $\sigma$ of the data. Note that for small $\sigma$, the data is easier to defend/harder to attack because everyone in the data set looks similar: the attacker cannot easily distinguish between the output distribution of the algorithm when the target $x \in D$ vs. when $x \notin D$. Conversely, large $\sigma$ makes it likely that some ``outlier'' $x$ that is easier for the attacker to identify will be in $X$. 
Thus, the ratio $\tilde{\eps}(X)/\eps$ increases with $\sigma$. On the other hand, the ratio $\tilde{\eps}(X)/\eps(X)$ does not significantly depend on $\sigma$. 

\begin{figure}[ht]
    \includegraphics[width=0.48\textwidth]{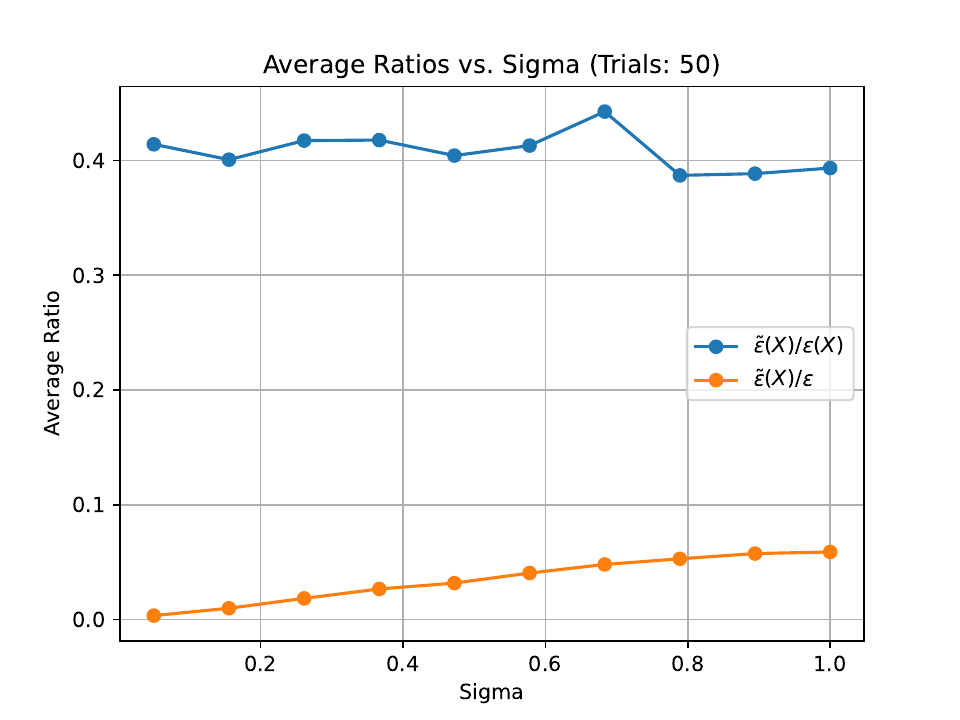}
    \caption{Ratios vs. $\sigma$, with $1$-dim. data, $n=6$, $m=10$, $C = 10$, $\eps(X) = 5$.}
    \label{fig: exp mech ratio v sigma}
\end{figure}

For example, when $\sigma = 1$ (standard normal data), the ratio $\tilde{\eps}(X)/\eps \approx 0.075$, which mostly reflects the fact that this data set is far from worst case. In this case, $\eps(X) = 5$ and $\eps \approx 28.5$, which does not afford any meaningful privacy guarantees under classical DP theory. However, the PMP parameter $\tilde{\eps}(X) \approx 2.14$, which provides a meaningful guarantee against practical MIAs on this particular subpopulation $X$, by Lemma~\ref{lem: attack success rate}. Moreover, for small $\sigma < 1$, the smaller ratios imply stronger PMP guarantees for fixed values of $\eps$: e.g. for $\sigma < .1$, the PMP parameter $\tilde{\eps}(X)$ approaches zero.

\begin{figure}[ht]
    \includegraphics[width=0.48\textwidth]{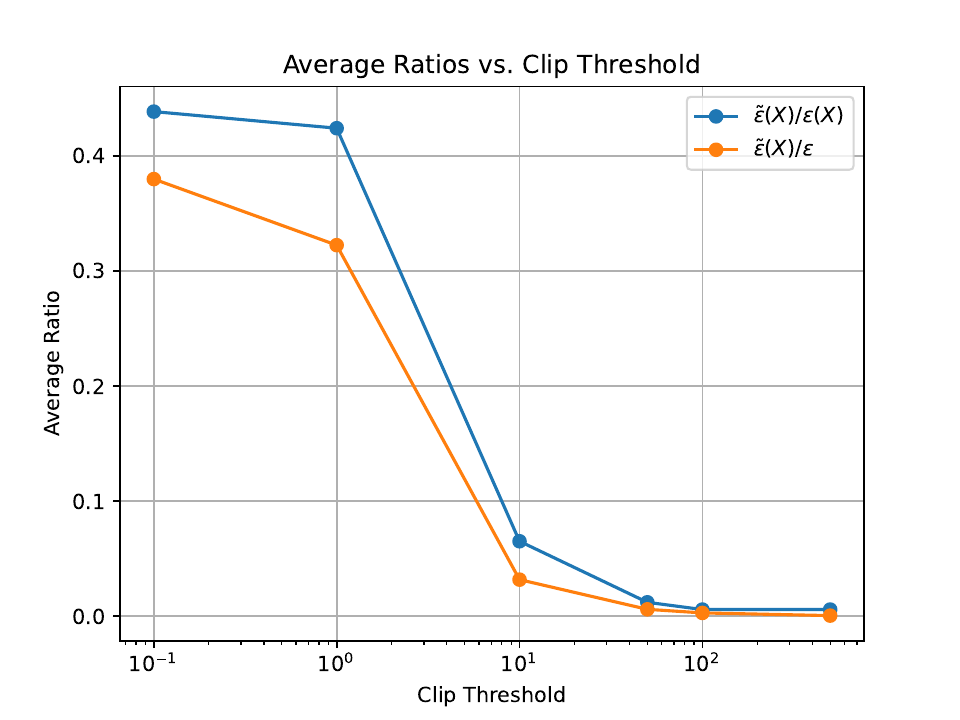}
    \caption{Ratios vs. Clip threshold $C$, with $5$-dim. data, $n=6$, $2$ outliers, $m=32$, $\sigma = 1$, $\eps(X) = 10$.}
    \label{fig: exp mech ratio v clip}
\end{figure}

Figure~\ref{fig: exp mech ratio v clip} shows the effect of clip threshold on the ratios. A small clip threshold $C$ reduces the effect of outlier data points, while a large clip threshold $C$ permits more outliers in the data set. To amplify the effect of outliers, we choose $2$ points in $X$ at random and multiply them by $100$. These extreme outliers cause $\eps(X)$ (and $\eps$) to be much larger than $\tilde{\eps}(X)$, since the worst-case DP attacker who knows an outlier in $D \setminus \{x\}$ can use this information to easily infer membership of $x$. By contrast, the practical attacker cannot use outliers to launch an MIA as effectively because they are uncertain about which other points are in $D$. For example, when $C = 50$, both ratios are less than $0.0123$. This means that a $10$-DP algorithm with no meaningful privacy guarantee against a worst-case attacker satisfies $0.123$-PMP and hence can defend against any practical attacker almost perfectly ($1/(1 + e^{-.123}) \approx 0.53$). Moreover, the DP parameter $\eps \approx \eps(X)$ in the presence of extreme outliers because such $X$ is nearly worst-case from a privacy perspective. In this experiment, $n=6$ since the runtime of computing $\tilde{\eps}(X)$ is exponential in $n$. We would expect the ratios to become even smaller for larger $n$ because the practical attacker's uncertainty would increase.

\begin{figure}[ht]
    \includegraphics[width=0.48\textwidth]{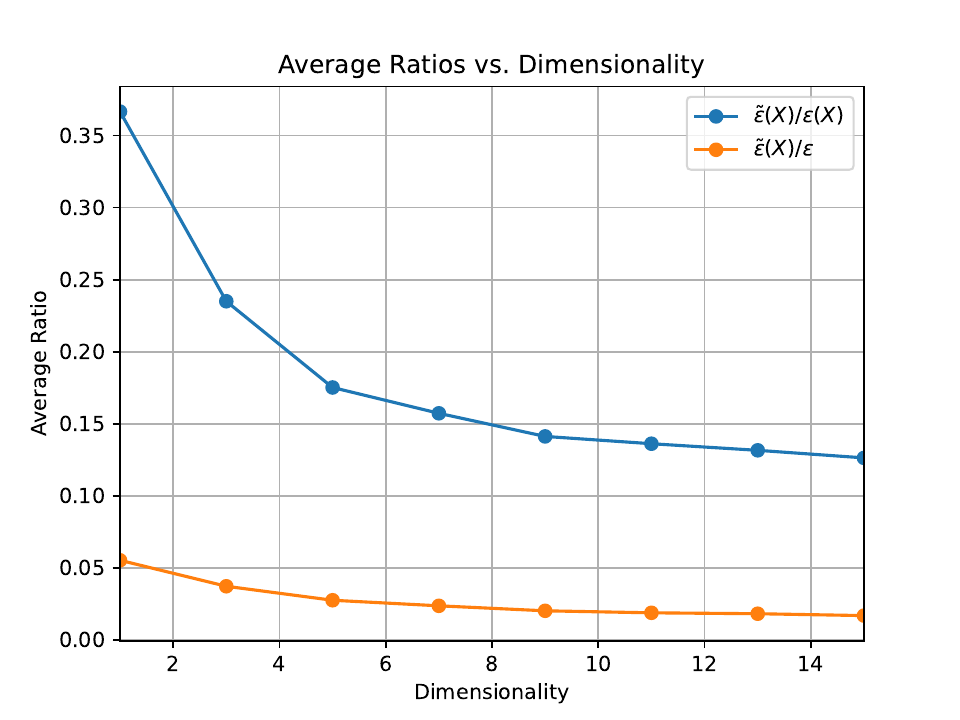}
    \caption{Ratios vs. Dimension of data $d$, with $n=6$, $m=10$, $\sigma = 1$, $\eps(X) = 2$.}
    \label{fig: exp mech ratio v dim}
\end{figure}

Finally, Figure~\ref{fig: exp mech ratio v dim} shows that the ratios become smaller as the dimension of the data increases. This can be attributed to the particular choice of loss function and Euclidean geometry in higher dimensions. In general, the effect of dimension on the ratios will depend on the loss function/problem.     

\section{Practical Privacy Guarantees of the Gaussian Mechanism}

This section analyzes the practical membership privacy of one of the most widely used $(\eps, \delta)$-DP algorithms: the \textit{Gaussian Mechanism}. Given a function $q: \XX^n \to \mathbb{R}^d$, the Gaussian mechanism simply adds isotropic Gaussian noise to the output of $q$: 
\[
\alg_G(D) := q(D) + \mathcal{N}\left(0, \sigma^2 \mathbf{I}_d\right).
\]

Denote the cumulative distribution function of $Y \sim \mathcal{N}(0, 1)$ by $\Phi$.

\begin{lemma}\cite{balle2018improving}
\label{lem: balle}
Let $q: \XX^n \to \mathbb{R}^d$ be a function with global $\ell_2$-sensitivity $\Delta = \sup_{D \sim D'} \|q(D) - q(D') \|_2$. For any $\eps \geq 0$ and $\delta \in [0, 1]$, the Gaussian mechanism $\alg_G$ is $(\eps, \delta)$-DP if and only if \begin{align}
\label{eq: balle}
\Phi\left(\frac{\Delta}{2\sigma} - \frac{\eps \sigma}{\Delta}\right) - e^{\eps} \Phi\left(-\frac{\Delta}{2\sigma} - \frac{\eps \sigma}{\Delta}\right) \leq \delta. 
\end{align}
\end{lemma}

For $\alg := \alg_G$ and $x \in X$, define the mixture distributions $\prin(S) := \frac{1}{|\xin|} \sum_{D \in \xin} \pr_{\alg}(\alg(D) \in S)$ and $\prout(S) := \frac{1}{|\xout|} \sum_{D \in \xout} \pr_{\alg}(\alg(D) \in S)$.
Our analysis will utilize the following characterizations of DP and PMP, which are immediate from the definitions:
\begin{lemma}
\label{lem: hockey}
Denote the \textit{hockey-stick divergence} between random variables $P$ and $Q$ by
$D_{e^{\eps}}(P \| Q):= \int_{\mathbb{R}} \max\{0, p(t) - e^{\eps} q(t)\}dt$, where $p$ and $q$ denote the probability density or mass functions of $P$ and $Q$ respectively. Then, $\alg$ is $(\eps, \delta)$-DP if and only if $\max\{ D_{e^{\eps}}(\alg(D) \| \alg(D')), D_{e^{\eps}}(\alg(D') \| \alg(D))\} \leq \delta$ for all $D \sim D'$. Moreover, $\alg$ is $(\eps, \delta)$-PMP w.r.t. $X$ if and only if $\max\{ D_{e^{\eps}}(\prin \| \prout), D_{e^{\eps}}(\prout \| \prin)\} \leq \delta$ for all $x \in X$.  
\end{lemma}

The following technical result will be crucial in our analysis.
\begin{proposition}
\label{prop: Gauss}
Let $\alg_G$ be the $(\eps, \delta)$-DP Gaussian mechanism. Then, for any $x \in X$,
\begin{align}
\label{eq: prop12 thing}
&\max\left\{\hockey(\prin || \prout), \hockey(\prout || \prin)\right\} \\
&\quad \leq \frac{1}{n |\xin|} \sum_{D \in \xin} \sum_{D' \in \xout, D' \sim D} \nonumber \\
&\quad \quad \Bigg[\Phi\left(\frac{\|q(D) - q(D')\|}{2 \sigma} - \frac{\eps \sigma}{\|q(D) - q(D')\|} \right) \nonumber \\
&\quad \quad - e^\eps\Phi\left(-\frac{\|q(D) - q(D')\|}{2 \sigma} - \frac{\eps \sigma}{\|q(D) - q(D')\|} \right) \Bigg].
\nonumber
\end{align} 
\end{proposition}

The main tools used in the proof of Proposition~\ref{prop: Gauss} are joint convexity of the hockey-stick divergence (which holds since $\hockey$ is an $f$-divergence) and a bound on $\hockey(\alg_G(D) || \alg_G(D')$ due to \cite{balle2018improving}.  

By Proposition~\ref{prop: Gauss} and Lemma~\ref{lem: hockey}, $\alg_G$ is $(\eps, \delta)$-PMP if the right-hand side of inequality~\ref{eq: prop12 thing} is upper-bounded by $\delta$. The differences between this sufficient condition for PMP and the condition~\cref{eq: balle} for DP is that~\cref{eq: balle} is worst-case over all pairs of adjacent data sets in $\XX^n$, whereas PMP only requires an average-case bound over all adjacent subsets of $X$.

\paragraph{Our Approach}
Our approach for analyzing the PMP parameter $\tilde{\eps}(X)$ for the $(\eps, \delta)$-DP Gaussian mechanism is as follows: \begin{enumerate}
    \item Given target DP parameters $(\eps, \delta)$, find the approximately smallest $\sigma$ such that the Gaussian mechanism is $(\eps, \delta)$ via Lemma~\ref{lem: balle} and ~\cite[Algorithm 1]{balle2018improving}. 
    \item Upper bound the hockey-stick divergence between $\prin$ and $\prout$ in Proposition~\ref{prop: Gauss}. 
    \item Using the value of $\sigma$ obtained in step 1), find the approximately smallest $\tilde{\eps}(X)$ such that our upper bound in Proposition~\ref{prop: Gauss} is $\leq \delta$ for all $x \in X$: this ensures that the Gaussian mechanism is $(\tilde{\eps}(X), \delta)$-PMP w.r.t. $X$, by Lemma~\ref{lem: hockey}. 
\end{enumerate}
Note that a naive implementation of step 3 would run in exponential (in $n$) time. To execute step 3 efficiently, we observe that the right-hand-side of Inequality~\ref{eq: prop12 thing} 
can be greatly simplified when the function is of the form $q(D) = \sum_{x \in D} f(x)$, where $f$ is some sample-wise function.
Since, the summation is constrained to be over $D$ containing $x$ and $D'$ that is adjacent to $D$, where $x$ is replaced with a different $x'$, the value of $q(D) - q(D')$ is equal to $f(x) - f(x')$.
Thus, the terms of the summation are a function of only $x'$ (given that $x$ is fixed), with each possible $x' \neq x$ repeatedly appearing an equal number of times.
Hence, instead of dealing with the average-case over all adjacent datasets, we can compute an equivalent average over all choices of $x' \neq x$, given by
\begin{align*}
&\frac{1}{2n - 1} \sum_{x' \neq x} \Bigg[\Phi\left(\frac{\|f(x) - f(x')\|}{2 \sigma} - \frac{\eps \sigma}{\|f(x) - f(x')\|} \right) \\
&\quad - e^\eps\Phi\left(-\frac{\|f(x) - f(x')\|}{2 \sigma} - \frac{\eps \sigma}{\|f(x) - f(x')\|} \right) \Bigg].
\end{align*} 

\paragraph{Numerical Simulations}

For our simulations, we consider empirical mean estimation: $q(D) = \sum_{x \in D} x/n$. The goals of these simulations are the same as in the simulations of the previous section: to quantify the ratios $\tilde{\eps}(X)/\eps$ and $\tilde{\eps}(X)/\eps(X)$ and understand the factors that cause these ratios to be large or small. We draw an i.i.d. Gaussian data set $X \sim \mathcal{N}(0, \sigma^2)^{2n \times d}$ and clip the $\ell_2$ norm of each data point, so $\|x_i\|_2 \leq C$ for all $i \in [2n]$, in order to bound global sensitivity of $q$.

Figure~\ref{fig:gauss mech ratio v epsX} shows the ratios vs. the DP parameter $\eps(X)$. First, note that the ratio $\tilde{\eps}(X)/\eps$ is small for all values of $\eps(X)$. For example, even when $\eps(X) = 10$ and $\tilde{\eps}(X)/\eps$ is at its largest, we still have a small PMP parameter $\tilde{\eps}(X) < 0.9$. Second, we see that there is a large gap between the two (orange and blue) curves, especially when $\eps(X)$ is large. This indicates that the worst-case DP parameter $\eps$ is significantly bigger than the subpopulation-specific DP parameter $\eps(X)$ in this experiment. Thus, $X$ is far from being worst-case. Third, the ratios increase with the DP parameter $\eps(X)$.  

\begin{figure}[ht]
    \includegraphics[width=0.48\textwidth]{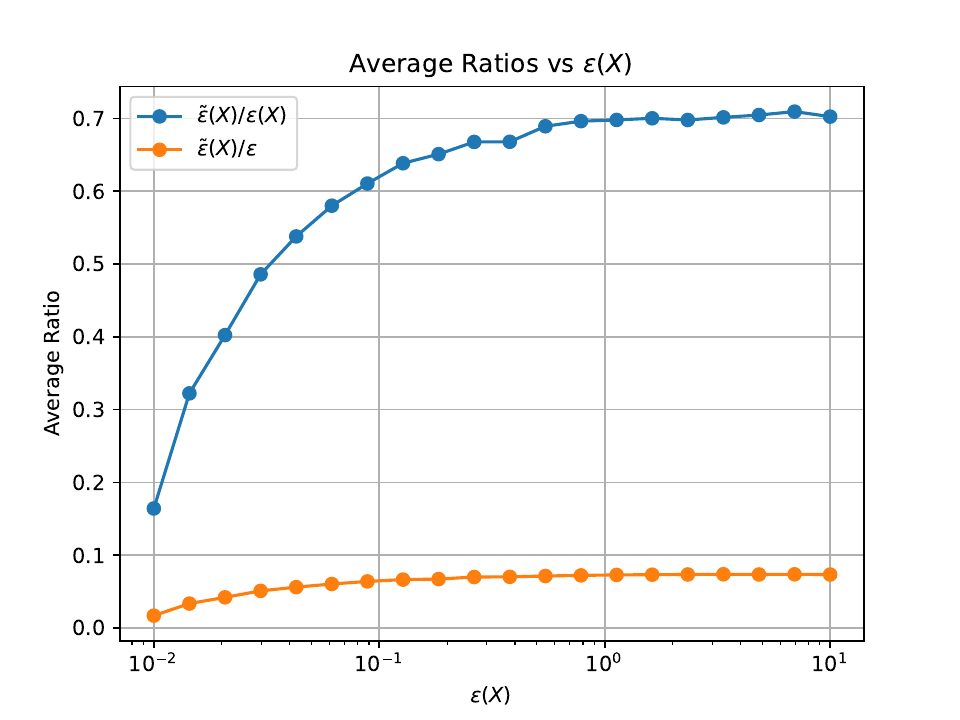}
    \caption{Ratios vs. $\eps(X)$, with $n=100$, $d = 20$, $C=50$, no outliers, $\sigma = 1$, $\delta = 10^{-2}$.}
    \label{fig:gauss mech ratio v epsX}
\end{figure}

Figure~\ref{fig:gauss mech ratio v clip} shows the effect of the clip threshold $C$ on the ratios in the presence of outliers. We produce outliers by choosing 2 points at random and scaling them by a factor of 10. Similar to Figure~\ref{fig: exp mech ratio v clip}, we see that the ratios shrink as the clip threshold $C$ increases. For example, for large $C = 100$, a DP parameter of $\eps = 5$ would translate into a much smaller PMP parameter of $\tilde{\eps}(X) = 1$. One difference between Figure~\ref{fig:gauss mech ratio v clip} and Figure~\ref{fig: exp mech ratio v clip} is that the gap between the blue and orange curves is larger in Figure~\ref{fig:gauss mech ratio v clip} than in Figure~\ref{fig: exp mech ratio v clip}. The reason is that the data $X$ is relatively easier to keep private in the experiment that was used to produce~\ref{fig:gauss mech ratio v clip}, whereas $X$ was nearly worst-case in Figure~\ref{fig: exp mech ratio v clip}. This is due to differences in the outlier scaling, dimension, $\sigma$, and the loss function/learning problem. 

\begin{figure}[ht]
    \includegraphics[width=0.48\textwidth]{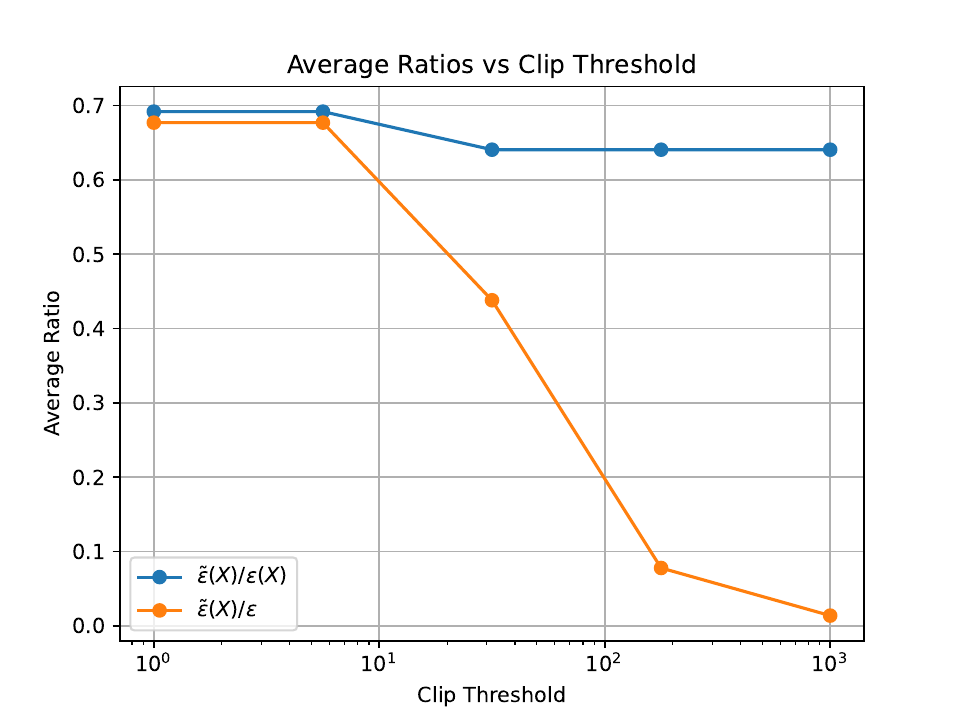}
    \caption{Ratios vs. Clip threshold $C$, with $n = 100$, $d = 10$, 2 outliers, $\sigma = 5$, $\delta = 10^{-2}$.}
    \label{fig:gauss mech ratio v clip}
\end{figure}

Figure~\ref{fig:gauss mech ratio v dims} shows that the ratios increase with the dimension of the data. In combination with Figure~\ref{fig: exp mech ratio v dim}, we see that the effect of dimension on the ratios may differ substantially for different learning problems. Thus, practitioners may want to apply problem-specific context to guide the choice of $\eps$. 

\begin{figure}[ht]
\includegraphics[width=0.48\textwidth]{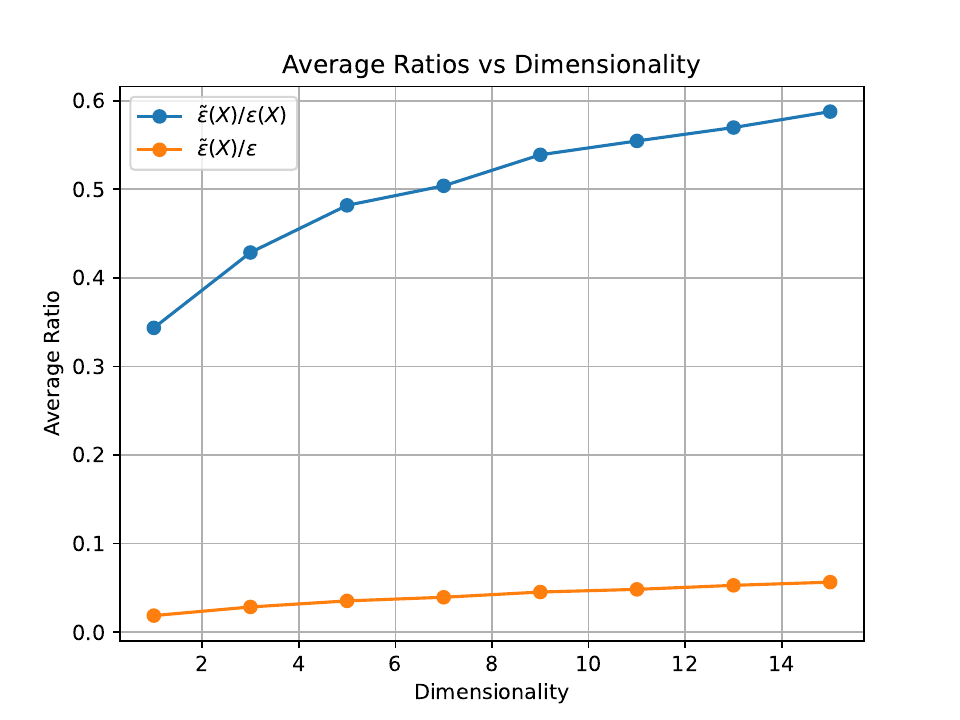}
    \caption{Ratios vs. Dimensionality, with $\eps$ with $n=100$, $C=50$, no outliers, $\sigma = 1$, $\delta = 10^{-2}$.}
    \label{fig:gauss mech ratio v dims}
\end{figure}

\section{Discussion and Conclusion}
In this paper, we analyzed the risk of data leakage of DP algorithms against a practical attacker who lacks certainty about the contents of the data set. At a high level, our results are encouraging: we rigorously show that even at larger $\eps$, DP mechanisms can actually provide guaranteed defense against practical MIAs.

We also gleaned more granular insights. For example, Figure~\ref{fig: exp mech ratio v sigma} indicates that if a data analyst has \textit{a priori} knowledge that the subpopulation from which data is drawn is approximately i.i.d./homogeneous, then they can afford to choose larger $\eps$: homogeneous data is easier to keep private. Also, data sets containing extreme outliers make it relatively much easier for a worst-case MIA to attack than for a practical MIA (e.g., see Figure~\ref{fig: exp mech ratio v clip}). Strategies like aggressive clipping can be used to mitigate the negative effects of outliers on privacy. Practitioners can use our code (which we plan to make available online) to help choose an appropriate $\eps$ for their particular problem/data population, while aiming to get a small corresponding PMP parameter, e.g., $\tilde{\eps}(X) \leq 0.1$.
 
We emphasize that our motivation for studying the notion of PMP was to better understand DP; we do not advocate for using PMP as a substitute for DP. PMP has certain shortcomings: As discussed in~\citet[Section 7]{ghazi2022algorithms}, an attacker's level of uncertainty may decrease over time, e.g., due to subsequent releases of information. Consequently, PMP does not satisfy the same sequential composition property that DP satisfies. We hope that by providing clearer interpretations of the DP parameter in terms of vulnerability to practical MIAs, our work facilitates more widespread use of DP algorithms in industry and government.

\bibliography{references}

\appendix 
\onecolumn 
\section{More details on related works}
\label{app: related work}
In this Appendix, we discuss prior works seeking to weaken assumptions about the attacker's knowledge in more detail. We highlight pathologies with previously proposed definitions, in which algorithms that clearly leak an individual's data can still satisfy these other definitions. Also, in contrast to some other works, PMP does not impose any distributional or independence assumptions on the underlying data. Instead, we allow for data to be drawn from an arbitrary subpopulation $X$. This makes our analysis harder, but also makes our definition and results stronger. Finally, we re-iterate that prior works did not provide the quantitative interpretations of practical privacy guarantees of concrete DP mechanisms that our work provides. 
\vspace{.2cm}

The work of~\citet{bassily2013coupled} was motivated by similar goals to our own.
They propose \textit{distributional DP} (DDP), a special case of their more general ``coupled-worlds privacy'' framework. DDP utilizes a ``simulator'' in its definition, requiring that the output distribution of the algorithm $\alg(D)$ be $(\eps, \delta)$-indistinguishable from the output distribution of some simulator run on the ``scrubbed'' data $\textit{Sim}(D_{-i})$, for all $i \in [n]$. 
The paper shows that certain noiseless protocols (e.g. real-valued summation, histograms, stable functions) satisfy DDP w.r.t. certain distribution classes $\PP$.
Moreover, they discuss the relation of DDP to previous notions of privacy---namely, Pufferfish privacy~\cite{kifer2012pufferfish} and noiseless privacy~\cite{bhaskar2011noiseless}.

However, the DDP definition suffers from a shortcoming, which does not occur with our PMP definition.
Under fairly mild distribution classes, DDP can permit pathological algorithms that simply release the entire dataset.
Let $\PP$ consist of distributions on datasets $D$ of size $n$, where knowledge of any $n-1$ points reveals the remaining point.
In such a case, the algorithm $\alg(D) = D$ becomes permissible, as it is perfectly ($\epsilon = \delta = 0$) indistinguishable from a simulator $\textit{Sim}(D_{-i})$ that can also output the entire dataset, by recovering the missing point.
For example, let $D_1, \ldots, D_n$ be uniformly drawn from binary sequences with even parity, i.e., for any $i \in [n]$, we have $D_i = \oplus_{j \in [n]: j \neq i} D_j$.
Note that this is a simple modification of Example~1 from~\citet{bassily2013coupled}, but with the auxiliary side information $Z$ (meant to reveal the parity of the binary sequence) omitted, and instead the fixed parity is incorporated into the distribution of the binary sequence.
This also applies to generalizations of this example with distributions where the sum (or mean) of the dataset is fixed and known.

\vspace{.2cm}
\citet{li2013membership} proposes a notion of membership privacy that is similar in spirit to DDP. Roughly speaking, an algorithm satisfies the (positive) membership privacy notion of \cite{li2013membership} w.r.t. a family of distributions $\PP$ on $\XX^n$ if $\mathbb{P}_{P, \alg}(x \in X | \alg(X) \in S) \approx \mathbb{P}_P(x \in X)$ for all $P \in \PP, x \in \XX, S \subset \mathbb{A}$. Conceptually, the essential differences  between this definition and our definition of PMP are: that our definition is parameterized by a parent data set, whereas theirs is parameterized by a family of distributions that correspond to the attacker's prior knowledge; also, their definition is non-symmetric in positive vs. negative membership inference, whereas our definition is symmetric. To address this latter limitation, \cite{li2013membership} introduce a second definition of negative membership privacy that protects against attacks that determine that someone was \textit{not} a member of the training data. Having two definitions seems unnecessary and our framework eliminates this need. They prove post-processing property of their membership privacy notion. The paper concludes by giving different instantiations of membership privacy for different choices of $\PP$ and recovering prior notions of privacy (including DP) along the way. In particular, \citet[Theorem 5.10]{li2013membership} shows that their membership privacy definition recovers ``DP under sampling''~\cite{li2012sampling} for the distribution family $\PP_{\beta}$ consisting of distributions such that $P(x) \in \{0, \beta\}$ for some choice of $\beta$. An algorithm satisfies ``DP under sampling'' if it is DP when composed with the subsampling operation that first samples each point in the data set with probability $\beta$ and then executes the algorithm on the subsampled data set. A drawback of \cite{li2013membership} is that the privacy parameter of their definition is not analyzed carefully or related to the DP parameter. We address this drawback in our work. 
\vspace{.2cm}

The work of~\citet{long2017towards} proposes \textit{differential training privacy} (DTP) to empirically estimate the privacy risk of publishing a classifier. Their DTP definition is specifically given for classifiers that output a vector of probabilities for predicted labels $y$ and features $x$: $p_{\alg(T)}(y | x)$, where $T$ is the training data set. Thus, their DTP notion is also data set-specific. Essentially, their definition requires that the predicted label probabilities of $\alg$ do not change too much when any single point in the training data set is removed: $p_{\alg(T)}(y | x) \leq e^\eps p_{\alg(T \setminus z)}(y|x)$ should hold for all $z \in T$ and all feature-label pairs $(x,y)$ in the universe.
Thus, their definition seems to be conceptually more similar to DP than it is to our definition of PMP. They provide an efficiently computable approximation of DTP that they compute in empirical case studies. They use these case studies to reason about the privacy risks of non-DP classifiers trained on certain data sets. No theoretical treatment of their DTP notion is provided.
\vspace{.2cm}

The work of~\citet{yeom2018privacy} proposes a different distribution-dependent definition of membership privacy based on the following membership experiment: data $S \sim P^n$ is drawn i.i.d. from some distribution and a learning algorithm $\alg(S)$ is run on the training data. Then a random bit $b \sim Ber(1/2)$ is drawn. If $b = 0$, then we draw a point $z \in S$ at random. If $b = 1$, then we draw a random point $z \sim P$. The attacker observes the target point $z$ and the output of the algorithm $\alg(S)$ (and implicitly has knowledge of $P$) and tries to guess the value of $b$ (i.e., membership of $z$). They define the membership advantage of an attacker in terms of its success rate, and say an algorithm is membership private (w.r.t. $P$) if every attacker has small membership advantage. Note that this membership experiment is the one that~\citet{carlini2022membership} assume in their attack model.
Compared to our PMP notion, a critical difference is that their definition only protects the privacy of the people in the data set \textit{on average} (over the random draw of $z$ from $S$). By contrast, our definition provides a stronger \textit{worst-case} (over $z \in S$) guarantee, ensuring that the data of \textit{every person} in $S$ remains private. Another difference is that \citet{yeom2018privacy} uses a parent distribution $P$, whereas we use a parent data set $X$.  
\citet{yeom2018privacy}'s definition is conceptually similar to DDP, but the precise way it is measured (in terms of advantage) differs and also it is framed as an experiment with an MIA. 

\citet{yeom2018privacy} shows that DP implies bounded membership advantage and studies the connection between overfitting and membership advantage. Additionally, they look at the connection between membership inference and attribute inference.

\vspace{.2cm}
The work of~\citet{humphries2020investigating} proposed a variation of the definition in \cite{yeom2018privacy} to deal with a specific limitation of \cite{yeom2018privacy}'s definition. Namely, \citet{humphries2020investigating} argues that the i.i.d. data assumption is problematic because DP guarantees become much weaker in the presence of data dependencies and because the assumption may not be satisfied in practice. Thus, they modify the definition in \citet{yeom2018privacy} by assuming that $P$ is a mixture of $K$ distributions: first, $k \sim [K]$ is drawn uniformly and then $S \sim P_k^n$ is drawn (conditionally i.i.d. given $k$). If $b=1$, then the target point $z$ is drawn from the mixture distribution: first $k' \sim [K]$ is drawn and then $z \sim P_{k'}$. Note that this modification allows for data dependencies. 

\citet{humphries2020investigating} provides tighter bounds on the relation between DP and membership advantage, compared with \cite{yeom2018privacy}. They also empirically evaluate membership inference with data dependencies. Again, the main difference between our notion and \citet{humphries2020investigating} is that we use a parent set instead of a parent distribution. Note that our definition also permits data dependencies, since the parent data set may consist of dependent data.

\vspace{.2cm}
The work of~\citet{sablayrolles2019white} defines a training algorithm that returns a parameter $\theta$ as being $(\eps, \delta)$-membership private w.r.t. a loss function $\ell(\theta, z)$. Their definition~\citet[Definition 3]{sablayrolles2019white} essentially requires a membership private algorithm to satisfy $\ell(\theta, z_1) \approx \int_w \ell(t, z_1) p_T(w) dw$ with high probability over the random draw of the training data set $T = (z_1, \ldots, z_n)$. Here $p_T(w)$ is the posterior density of the parameter $w$ given $(z_2, \ldots, z_n)$, which is assumed to take a particular form given in \cite[Definition 12]{sablayrolles2019white}. Roughly speaking, it is assumed that $w$ depends on the data through an ``exponential mechanism''-like training algorithm. An immediate problem with their definition is the dependence on the loss function, which greatly reduces the generality and flexibility of the definition. 
\cite{sablayrolles2019white} characterize the optimal MIA under certain assumptions discussed above. They show that DP implies a bound on the membership advantage. They run experiments showing that their attack---based on the theoretically optimal attack under their assumption on the posterior---performs well.

\vspace{.2cm}
The work of~\citet{mahloujifar2022optimal} is motivated by the desire to get tighter bounds on membership inference privacy for existing algorithms. They measure membership inference privacy by using a very strong definition of membership privacy that is similar to DP in that it assumes (implicitly) that the attacker knows the other $n-1$ points in the training set. As we argue, this assumption is usually unrealistic and a major benefit of our PMP definition is that it relaxes this assumption by modeling the adversary's uncertainty about the training data set.

\vspace{.2cm}
The recent work of~\citet{izzo2022provable} works towards a theory of membership inference privacy (MIP). Their notion of $\eta$-MIP is similar to our notion of $\eps$-PMP in terms of being average case over the uniformly random draw of the training data, and worst-case over outcomes. However, their MIP notion is fundamentally weaker than our PMP notion. In particular, it is easy to see that the following blatantly non-private algorithm satisfies $\eta$-MIP but does not satisfy $\eps$-PMP for any $\eps < \infty$: $\alg$ releases a training example $D_1$ with probability $n\eta$ and otherwise outputs \textit{NULL}. Thus, their MIP notion may not be strong enough to offer the meaningful and intuitive membership privacy guarantees that we desire. Moreover, PMP implies MIP, as the following lemma shows:
\begin{lemma}
\label{lem: PMP implies izzo}
If $\alg$ is $\eps$-PMP, then $\alg$ is $\frac{1 - e^{-\eps}}{2}$-MIP. 
\end{lemma}
\begin{proof}
Let $\Delta = 1 - e^{-\eps} \in [0, 1).$ Assume for concreteness that $\alg$ is discrete. (A similar argument works if $\alg$ is continuous.) Then since $\alg$ is $\eps$-PMP, we have \[
1 - \Delta \leq \frac{\pr(x \in D | \alg(D) = a)}{\pr(x \notin D | \alg(D) = a)} \leq \frac{1}{1 - \Delta}
\] 
for almost every $a \in \ZZ$, where $\ZZ$ denotes the range of $\alg(D)$. By the proof of \cite[Theorem 7]{izzo2022provable}, we get \[
\max\left(\pr(x \in D | \alg(D) = a), \pr(x \notin D | \alg(D) = a)\right) \leq \frac{1+\Delta}{2}
\]
and 
\begin{align*}
\int_{\ZZ} \max\left(\pr(x \in D | \alg(D) = a), \pr(x \notin D | \alg(D) = a)\right) \pr(\alg(D) = a)
&\leq \frac{1 + \Delta}{2} \int_{\ZZ} P(\alg(D) = a) \\
&= \frac{1 + \Delta}{2}.
\end{align*}
By the definition of $\eta$-MIP, the above inequality implies that $\alg$ is $\eta$-MIP for $\eta = \Delta/2 = (1 - e^{-\eps})/2$.
\end{proof}

Finally, the concurrent and independent work of~\citet{leemann2023gaussian} proposes a Gaussian-DP analog of the membership inference privacy (MIP) notion. They show how to implement their Gaussian MIP with noisy SGD and give a novel MIA based on their MIP notion.

\section{Proofs of Theoretical Results}
In this Appendix, we re-state and prove our theoretical results. First, we show that PMP satisfies post-processing. 
\begin{lemma}[Post-processing property of PMP]
Let $\alg: \XX^n \to \ZZ$ be $(\eps, \delta)$-PMP. If $f: \ZZ \to \mathcal{Y}$ is any function, then $f \circ \alg: \XX^n \to \mathcal{Y}$ is $(\eps, \delta)$-PMP. 
\end{lemma}
\begin{proof}
Let $S \subset \mathcal{Y}$ be measurable, $X \in \XX^{2n}$, $x \in X$, $N = |\xin| = |\xout|$, where $\xin$ and $\xout$ are defined in Lemma~\ref{lem: equivalent definitions of PMP}. Assume w.l.o.g. that $f$ is deterministic. (If $f$ is randomized, then we can reduce to the deterministic case by considering convex combinations.) Let $T_S := \{z \in \ZZ : f(z) \in S\} = f^{-1}(S)$. Note that for any $D \in \xin$, there exists a $D' \in \xout$ that is adjacent to $D$: if $x = D_i$, take $D' = (D_1, \ldots, D_{i-1}, x', D_{i+1}, \ldots, D_n)$ for some $x' \in X \setminus D$. 
Then, by Lemma~\ref{lem: equivalent definitions of PMP}, we have
\begin{align*}
\frac{1}{N}\sum_{D \in \xin} \pr_{\alg}(f \circ \alg(D) \in S) &= \frac{1}{N}\sum_{D \in \xin} \pr_{\alg}(\alg(D) \in T_S) \\
&\leq \frac{1}{N} \sum_{D' \in \xout} e^{\eps} \pr_{\alg}(\alg(D') \in T_S)  + \frac{\delta}{2}\\
&= \delta/2 + e^\eps \frac{1}{N}\sum_{D' \in \xout} \pr_{\alg}(f \circ \alg(D') \in S).
\end{align*}
By Lemma~\ref{lem: equivalent definitions of PMP}, we conclude that  $f \circ \alg$ is $(\eps, \delta)$-PMP. 
\end{proof}

\begin{lemma}
[Re-statement of Lemma~\ref{lem: equivalent definitions of PMP}]
Let $X \in \XX^{2n}$, $x \in X$, $\xin := \{D \subset X : |D| = n, x \in X\}$, and $\xout = \{D \subset X : |D| = n, x \notin X\}$. Let $S \subset \ZZ$ be a measurable set. 
If
\begin{align}
\label{eq: a app}
e^{-\eps}\left(\pr(x \notin D | \alg(D) \in S) - \delta\right) 
\leq \pr(x \in D | \alg(D) \in S)
\leq e^\eps \pr(x \notin D | \alg(D) \in S) + \delta,
\end{align}
then
\begin{align}
\label{eq: b app}
e^{-\eps}\left(\pr(\alg(D) \in S | x \notin D) - 2\delta\right)
\leq \pr(\alg(D) \in S | x \in D ) 
\leq e^\eps \pr( \alg(D) \in S |x \notin D ) + 2\delta.
\end{align}
Also, \cref{eq: b app} holds iff \begin{align}
\label{eq: c app}
e^{-\eps}\left(\frac{1}{N} \sum_{D' \in \xout} \pr_{\alg}(\alg(D') \in S) - \delta \right)
\leq \frac{1}{N} \sum_{D \in \xin} \pr_{\alg}(\alg(D) \in S)
\leq e^\eps \left(\frac{1}{N} \sum_{D' \in \xout} \pr_{\alg}(\alg(D') \in S)\right) + \delta,
\end{align}
where $N := |\xin| = |\xout| = {2n \choose n}/2$ and the probabilities in \cref{eq: c app} are taken solely over the randomness of $\alg$.

Moreover, if $\delta = 0$, then \cref{eq: a app} holds iff \cref{eq: b app} holds iff \cref{eq: c app} holds. Thus, $\alg$ is $\eps$-PMP w.r.t. $X$ iff any of these three inequalities holds for all $x \in X$ and all $S \subset \mathcal{Z}$. 
\end{lemma}
\begin{proof}
Suppose \cref{eq: a app} holds. Then, by Bayes' rule and the fact that $\pr(x \in D) = \pr(x\notin D) = 1/2$, we have \begin{align}
\label{eq: i}
    e^{-\eps}\left(-\delta + \frac{\pr(\alg(D) \in S | x\notin D)}{2 \pr(\alg(D) \in S)} \right) \leq \frac{\pr(\alg(D) \in S | x \in D)}{2 \pr(\alg(D) \in S)} \leq e^\eps \frac{\pr(\alg(D) \in S | x\notin D)}{2 \pr(\alg(D) \in S)} + \delta.    \end{align}
Multiplying \cref{eq: i} by $2\pr(\alg(D) \in S)$ and using the fact that $\pr(\alg(D) \in S) \in [0, 1]$ yields \cref{eq: b app}. 

Next we prove the equivalence between \cref{eq: b app} and \cref{eq: c app}. Observe that \begin{align}
    \pr(\alg(D) \in S | x \in D) &= \pr(\alg(D) \in S | D \in \xin) \\
    &= \frac{\pr(\alg(D) \in S, D \in \xin)}{\pr(D \in \xin)} \\
    &= \frac{\frac{1}{N} \sum_{D \in \xin} \pr_{\alg}(\alg(D) \in S)}{1/2} \\
    &= \frac{2}{N}\sum_{D \in \xin} \pr_{\alg}(\alg(D) \in S).
\end{align}
Similarly, $\pr(\alg(D) \in S | x\notin D) = \frac{2}{N}\sum_{D' \in \xout} \pr_{\alg}(\alg(D') \in S)$. Substituting these equalities into \cref{eq: b app} and then dividing by $2$ yields \cref{eq: c app}. 

Now suppose $\delta = 0$. Then we have already shown that \cref{eq: a app} implies \cref{eq: b app} and that \cref{eq: b app} is equivalent to \cref{eq: c app}. Conversely, if \cref{eq: b app} holds, then by Bayes rule and the fact that $\pr(x \in D) = \pr(x \notin D) = 1/2$, we get 
\begin{align*}
    e^{-\eps} 2 \pr(x \notin D | \alg(D) \in S) \pr(\alg(D) \in S) \leq 2 \pr(x \in D | \alg(D) \in S) \pr(\alg(D) \in S) \leq e^{\eps} 2 \pr(x \notin D | \alg(D) \in S) \pr(\alg(D) \in S).
\end{align*}
If $\pr(\alg(D) \in S) > 0$, then dividing the above by $2 \pr(\alg(D) \in S)$ implies that \cref{eq: a app} holds. This completes the proof. 
\end{proof}

\begin{corollary}[Re-statement of Corollary~\ref{coro: DP = PMP}]
If $n=1$, then $\alg$ is $(\eps, \delta)$-DP iff $\alg$ is $(\eps, 2 \delta)$-PMP w.r.t. $X$ for every $X \in \XX^{2n}$.    
\end{corollary}
\begin{proof}
If $n=1$, then $N=1$ and the sums in \cref{eq: c} are each only over one term. Thus, \cref{eq: c} holds for all $X = \{x, x'\} \in \XX^2$ iff \begin{align}
    e^{-\eps} (\pr(\alg(x') \in S) - \delta) \leq \pr(\alg(x) \in S) \leq e^{\eps} \pr(\alg(x') \in S) + \delta
\end{align}
iff $\alg$ is $(\eps, \delta)$-DP. By \cref{lem: equivalent definitions of PMP}, this condition is also equivalent to $\alg$ being $(\eps, 2\delta)$-PMP w.r.t. $X$ for every $X \in \XX^2$. 
\end{proof}

\begin{proposition}[Re-statement of Proposition~\ref{prop: PMP is weaker than DP}]
If $\alg$ is $\eps$-DP, then $\alg$ is $\eps$-PMP. Moreover, if $n > 2$, then there exists an $\ln(2)$-PMP $\alg$ that is not $\eps'$-DP for any $\eps' < \infty$. 
\end{proposition}
\begin{proof}
The first statement is a consequence of Lemma~\ref{lem: equivalent definitions of PMP} and uses arguments from the proof of~\cite[Proposition 6]{izzo2022provable}. Let $X \in \XX^{2n}$ consist of $2n$ distinct points, let $x \in X$ and $S \subset \ZZ$. Let $\alg$ be $\eps$-DP. By Lemma~\ref{lem: equivalent definitions of PMP}, we have \begin{align*}
\frac{\pr(\alg(D) \in S | x \notin D)}{\pr(\alg(D) \in S | x \in D)} &= \frac{\sum_{D \in \xin} \pr_{\alg}(\alg(D) \in S)}{\sum_{D' \in \xout} \pr_{\alg}(\alg(D') \in S)}.
\end{align*}
Now, for any $D = (D_1, \ldots, D_n) \in \xin$, there is a unique $i \in [n]$ such that $D_i = x$. Let $x' \in X \setminus D$ and $D' := (D_1, \ldots, D_{i-1}, x', D_{i+1}, \ldots, D_n)$, which is a neighboring data set of $D$ (i.e. $D \sim D'$) and $D' \in \xout$. Note that there are $n$ choices for $x'$. Thus, we can see that $D$ has $n$ neighboring data sets in $\xout$. Similarly, every $D' \in \xout$ has $n$ neighbors in $\xin$. Thus, \[
n \sum_{D' \in \xout} \pr_{\alg}(\alg(D') \in S) = \sum_{D \in \xin} \sum_{\substack{D' \in \xout \\ D' \sim D}} \pr_{\alg}(\alg(D') \in S).
\]
This implies \begin{align*}
\frac{\pr(\alg(D) \in S | x \notin D)}{\pr(\alg(D) \in S | x \in D)} &= \frac{\sum_{D \in \xin} \pr_{\alg}(\alg(D) \in S)}{\sum_{D' \in \xout} \pr_{\alg}(\alg(D') \in S)} \\
&=\frac{\sum_{D \in \xin} \pr_{\alg}(\alg(D) \in S)}{\frac{1}{n} \sum_{D \in \xin} \sum_{D' \in \xout, D' \sim D}\pr(\alg(D') \in S)} \\
&\leq \frac{e^\eps \sum_{D \in \xin} \min_{D' \in \xout, D' \sim D} \pr_{\alg}(\alg(D') \in S)}{\sum_{D \in \xin} \text{Average}_{D' \in \xout, D' \sim D} \pr_{\alg}(\alg(D') \in S)} \\
&\leq e^{\eps}. 
\end{align*}
A similar argument proves the other inequality. 

For the second statement, assume for simplicity that $n=3$. It will be easy to see that our construction extends to $n > 3$ (and indeed the PMP parameter can be reduced for $n > 3$, giving a stronger result). Let $\XX = \{0, 1, 2, 3, 4, 5\}$, and $X = (0, 1, 2, 3, 4, 5)$. 
Define $\alg(D) = \text{sum}(D)~(\text{mod}~6)$ as the modular addition operator. First, $\alg$ is clearly not $\eps'$-DP for any $\eps' < \infty$ since $\alg$ is not randomized. Concretely, if $\alg(D) = 0$ for some $D \in \XX^n$, then replacing $D_1$ by $D_1' = D_1 + 1~(\text{mod} 6)$ and letting $D' = (D_1', D_2, \ldots, D_n)$ implies that $\alg(D') = 1$; hence the privacy loss is infinite. 

Next, $\alg$ is $\ln(2)$-PMP. To see this, let $x = 0$ and compute $\max_{a \in \XX} \sum_{D \in \xin} \pr_{\alg}(\alg(D) = a) = \max_{a \in \XX} |\{D \in \xin : \sum(D) = a~(\text{mod}~6)\}| = 2$. On the other hand, $\min_{a \in \XX} \sum_{D \in \xout} \pr_{\alg}(\alg(D) = a) = \min_{a \in \XX} |\{D \in \xout : \sum(D) = a~(\text{mod}~6)\}| = 1$. 
By symmetry and Lemma~\ref{lem: equivalent definitions of PMP}, $\alg$ is $\eps$-PMP with respect to $X$ if and only if \[
\left|\ln \frac{\sum_{D \in \xin} \pr(\alg(D) = a)}{\sum_{D' \in \xout} \pr(\alg(D') = a)} \right| \leq \eps
\]
for all $a \in \XX$. By the above computations, we see that this holds only if $\eps \geq \ln(2)$. 
\end{proof}

\begin{lemma}[Re-statement of Lemma~\ref{lem: attack success rate}]
Let $\alg$ be $\eps$-PMP with respect to $X$ and $\mathcal{M}$ be any practical MIA. Then, the probability that $\mathcal{M}$ successfully infers membership, for any $x \in X$, never exceeds $1/(1 + e^{-\eps})$. 
\end{lemma}
\begin{proof}
Suppose $\pr(\alg(D) \in S | x\in D) = e^\eps \pr(\alg(D') \in S | x \notin D')$ for some $S \subset \ZZ$ and $x \in X$. Then any practical MIA's success probability is upper bounded by $\max\left(\pr(x \in D | \alg(D) \in S), \pr(x \notin D | \alg(D) \in S)\right)$, which corresponds to the success probability of the Bayes optimal practical MIA, $\mathcal{M}^*$. Assume w.l.o.g. that $\max\left(\pr(x \in D | \alg(D) \in S), \pr(x \notin D | \alg(D) \in S)\right) = \pr(x \in D | \alg(D) \in S)$.  Then, \begin{align*}
\pr(\mathcal{M}^* ~\text{is correct}) &\leq \pr(x \in D | \alg(D) \in S) \\
&= \frac{\pr(\alg(D) \in S | x \in D) \pr(x\in D)}{\pr(\alg(D) \in S)} \\
&= \frac{e^\eps \pr(\alg(D) \in S | x\notin D) \times 1/2}{(1/2) \left(\pr(\alg(D) \in S | x\in D) + \pr(\alg(D) \in S | x\notin D) \right)} \\
&= \frac{e^\eps}{1 + e^\eps} = \frac{1}{1 + e^{-\eps}}.
\end{align*}
\end{proof}

\begin{proposition}[Re-statement of Proposition~\ref{prop: exponential mechanism PMP}]
Let $X \in \XX^{2n}$. The $\eps$-DP exponential mechanism is $\tilde{\eps}(X)$-PMP with respect to $X$ if and only if \[
\tilde{\eps}(X) \geq \ln\left[\frac{\sum_{D \in \xin} c(D) \exp\left(-\frac{\eps}{2 \Delta_\ell} \ell(w,D) \right)}{\sum_{D' \in \xout} c(D') \exp\left(-\frac{\eps}{2 \Delta_\ell} \ell(w,D')\right)} \right]
\]
for all $w \in \WW$ and $x \in X$,
where $c(D) = \left[\sum_{w' \in \WW} \exp\left(\frac{- \eps \ell(w', D)}{2\Delta_{\ell}} \right)\right]^{-1}$ and $c(D')$ is defined similarly. 
\end{proposition}
\begin{proof}
By Lemma~\ref{prop: exponential mechanism PMP}, the $\eps$-DP exponential mechanism is $\tilde{\eps}(X)$-PMP with respect to $X$ if and only if \begin{align*}
\tilde{\eps} &\geq \max_{w \in \WW, x \in X} \ln\left[\frac{\sum_{D \in \xin} \pr_{\alg}(\alg(D) = w)}{\sum_{D' \in \xout} \pr_{\alg}(\alg(D') = w} \right] \\
&=\max_{w \in \WW, x \in X} \ln\left[\frac{\sum_{D \in \xin} c(D) \exp\left(-\frac{\eps}{2 \Delta_\ell} \ell(w,D) \right)}{\sum_{D' \in \xout} c(D') \exp\left(-\frac{\eps}{2 \Delta_\ell} \ell(w,D')\right)} \right].
\end{align*}
\end{proof}

\begin{proposition}[Re-statement of Proposition~\ref{prop: Gauss}]
Let $\alg_G$ be the $(\eps, \delta)$-DP Gaussian mechanism. Then, for any $x \in X$,
\begin{align*}
&\max\left\{\hockey(\prin || \prout), \hockey(\prout || \prin)\right\} \leq \frac{1}{n |\xin|} \sum_{D \in \xin} \sum_{\substack{D' \in \xout \\ D' \sim D}} \\
& \quad \quad
\Bigg[\Phi\left(\frac{\|q(D) - q(D')\|}{2 \sigma} - \frac{\eps \sigma}{\|q(D) - q(D')\|} \right) 
- e^\eps\Phi\left(-\frac{\|q(D) - q(D')\|}{2 \sigma} - \frac{\eps \sigma}{\|q(D) - q(D')\|} \right) \Bigg].
\end{align*} 
\end{proposition}

\begin{proof}
Let $x \in X$ and $N := |\xin| = |\xout|$. Let $P(S) := \pr_{\alg}(\alg(D) \in S)$, $P'(S) := \pr_{\alg}(\alg(D') \in S)$, and denote the density functions of these distributions by $p$ and $p'$ respectively. (The distributions $P$ and $P'$ are parameterized by specific data sets $D$ and $D'$, but we omit the dependence to reduce notational clutter.) Note that $p(t) = \frac{1}{\sqrt{2\pi \sigma^2}} \exp\left(- \frac{(q(D) - t)^2}{2\sigma^2} \right)$ and $p'(t) = \frac{1}{\sqrt{2\pi \sigma^2}} \exp\left(- \frac{(q(D') - t)^2}{2\sigma^2} \right)$.
Recall that the hockey-stick divergence $\hockey$ is an $f$-divergence, with $f(t) = f_\eps(t) = \max(t - e^{\eps}, 0)$ \cite{sason2016fdivergence}. By joint convexity, 
\begin{align}
\label{ineq: gauss}
\hockey(\prin || \prout) &\leq \frac{1}{N}\sum_{D \in \xin} \hockey\left(P \Bigg |\Bigg| \normalsize \frac{1}{n} \sum_{\substack{D' \in \xout \\ D' \sim D}} P' \right) \nonumber \\
&\leq \frac{1}{N} \sum_{D \in \xin} \frac{1}{n} \sum_{\substack{D' \in \xout \\ D' \sim D}} \hockey\left(P || P'\right). 
\end{align}
Now, \begin{align*}
\hockey(P || P') &= \int \max\left(0, q(t) - e^\eps q'(t) \right) dt \nonumber \\
&= \int_{t: q(t) \geq e^{\eps} q'(t)} [q(t) - e^{\eps} q'(t)] dt,
\end{align*}
and by \cite{balle2018improving}, we have \begin{align*}
\hockey(P || P') &= \pr_{y \sim \alg(D)|D}\left[\log \frac{p(y)}{p'(y)} > \eps \right] - e^{\eps} \pr_{z \sim \alg(D')|D'}\left[\log \frac{p'(y)}{p(y)} < -\eps \right] \\
&= \Phi\left(\frac{\|q(D) - q(D')\|}{2\sigma} - \frac{\eps \sigma}{\|q(D) - q(D')\|}\right) - e^\eps \Phi\left(- \frac{\|q(D) - q(D')\|}{2\sigma} - \frac{\eps \sigma}{\|q(D) - q(D') \|}\right).
\end{align*}
Plugging this identity into the inequality~\ref{ineq: gauss} completes the proof. 
\end{proof}

\end{document}